\documentclass[review]{elsarticle}

\usepackage{lineno,hyperref}
\usepackage[full]{textcomp} % to get the right copyright, etc.
\usepackage{fbb} % osf for text, lining for math
\usepackage[scaled=.95]{cabin}
\usepackage{url}
\usepackage[english]{babel}
\usepackage{graphics,epsfig}
\usepackage[varqu,varl]{inconsolata}% typewriter
\usepackage{color}
\usepackage{amsmath,amsthm,amsfonts,amssymb}
\newtheorem{theorem}{Theorem}
\newtheorem{proposition}{Proposition}
\newtheorem{example}{Example}

\newtheorem{lemma}{Lemma}

\usepackage{centernot}
\usepackage{epsfig}
\usepackage{graphicx}
\usepackage{pstricks}
\usepackage{float}
\usepackage{lineno}
\begin{document}
	\begin{frontmatter}
		\title{The Interval function, Ptolemaic, distance hereditary, bridged  graphs and  axiomatic characterizations}
		\author{Manoj Changat\fnref{mcthanks}}
		\address{Department of Futures Studies,\\University of Kerala, \\Thiruvananthapuram - 695 581, India \\email: mchangat@keralauniversity.ac.in}
		\author{Lekshmi Kamal K. Sheela\fnref{lkksthanks}}
		\address{Department of Futures Studies,\\University of Kerala, \\Thiruvananthapuram - 695 581, India \\email: lekshmisanthoshgr@gmail.com}
				\author{Prasanth G. Narasimha-Shenoi\fnref{pgnthanks}}
				\address{Department of Mathematics, \\Government College Chittur,\\Palakkad - 678 104, India\\email: prasanthgns@gmail.com}
\fntext[mcthanks]{Supported by Department  of  Science  and  Technology,  SERB  (File  No.MTR/2017/000238 “Axiomatics of Betweenness in Discrete Structures”).}
\fntext[lkksthanks]{Supported by CSIR,  Government India as a CSIR-JRF.}
\fntext[pgnthanks]{Supported by Department of Science and Technology,
	SERB under their MATRICS Scheme (File No.MTR/2018/000012).}
%\maketitle
  \begin{abstract}
In this paper we consider certain types of betweenness axioms on the interval function $I_G$ of a connected graph $G$.  
We characterize the class of graphs for which $I_G$  satisfy these axioms. The class of
graphs that we characterize include the important class of Ptolemaic graphs and some proper superclasses of Ptolemaic graphs: the distance hereditary graphs and the bridged graphs. We also provide axiomatic characterizations of the interval function of these classes of graphs using an arbitrary function known as \emph{transit function}. 
\end{abstract}
\begin{keyword}
	\texttt{transit function}\sep interval function\sep betweenness axioms\sep Ptolemaic graphs\sep distance hereditary graphs\sep bridged graphs
	\MSC[2020] 05C12 
\end{keyword}
\end{frontmatter}
%\noindent {\bf Keywords}: transit function, interval function, betweenness axioms, Ptolemaic graphs, distance hereditary graphs, bridged graphs
%\linenumbers
%\noindent {\bf AMS subject classification (2020)}: 05C12
%\begin{center}
%
%{\textbf{Axiomatic characterization of the interval function of bridged  graphs}}
%\end{center}

\section{Introduction} 

%Let $V$ be a finite, nonempty set with power set $2^V$. Throughout this paper $R$ denotes a function $R:V\times V\to 2^V$, for short, we say that $R$ is a function on $V$.
Transit functions on discrete structures were introduced by Mulder
\cite{muld-08} to generalize some basic notions in discrete geometry, amongst which convexity, interval and betweenness.
A \emph{transit function} on a non-empty set $V$ is a function $R: V\times V$ to $2^{V}$  on $V$ satisfying the following three axioms:

 \begin{description}
 \item [$(t1)$] $u \in R(u,v)$, for all $u,v \in V$,
 \item [$(t2)$] $R(u,v) = R(v,u)$, for all $u,v \in V$,
 \item [$(t3)$] $R(u,u) = \{u\}$, for all $u \in V$.
 \end{description}

If $V$ is the vertex set of a graph $G$, then we say that $R$ is a transit function on $G$. Throughout this paper, we consider only finite, simple and  connected graphs. 
%Axiom $(t3)$ was added mostly to exclude some degenerate cases.
The underlying graph $G_{R}$ of a transit function $R$ on $V$ is the graph with vertex
set $V$, where two distinct vertices $u$ and $v$ are joined by an edge if and
only if $R(u,v) = \{u,v\}$. %Note that, if $R$ is a transit function on $G$, then $G_R$ need not be isomorphic with $G$. 

A $u,v$ - \emph{shortest path} in a connected graph $G= (V,E)$ is a $u,v$-path in $G$ containing the minimum number of edges. The length of a shortest $u,v$-path $P$ (that is, the number of edges in $P$) is the standard distance in $G$.  The interval function $I_G$ of a connected graph $G$ is the function $I_G$  defined with respect to the standard distance $d$ in $G$ as  
 
 \noindent
 $I: V\times V :\longrightarrow 2^{V}$ 
 
\begin{center}
$I_G(u,v)=\{w\in V$: $w$ lies on some shortest $u,v$ - path in $G \} =\{w\in V: d(u,w) +d(w,v) = d(u,v) \}$
\end{center}

The interval function $I_G$ is a classical example of a transit function on a graph ( we some times denote $I_G$ by $I$, if there is no confusion for the graph $G$).  It is easy to observe that the underlying graph $G_{I_G}$ of $I_G$ is isomorphic to $G$. The term interval function was coined by Mulder in \cite{mu-80}, where it is extensively studied using an axiomatic approach. \\
Nebesk\'{y} initiated a very interesting problem on the interval function $I$ of a connected graph $G=(V, E)$ during the 1990s.  The problem is the following:  `` Is it possible to give a characterization of the interval function $I_G$ of a connected graph $G$ using a set of simple axioms (first - order axioms) defined on a transit function $R$ on $V$?''
Nebesk\'{y} \cite{nebe-94,nebesky-94} proved that there exists such a characterization for the interval function $I(u,v)$ by using first - order axioms on an arbitrary  transit function $R$. In further papers that followed \cite{nebe-95,ne-08,nebesky-08,nebe-01}, Nebesk\'{y} improved the formulation and the proof of this characterization.  Also, refer Mulder and Nebesk\'{y} \cite{mune-09}.  In \cite{chfermuna-18}, the axiomatic characterization of $I_G$  is extended to that of a disconnected graph. In all these characterizations,  five essential axioms known as classical axioms are always required. These five classical axioms are $(t1)$ and $(t2)$ and three additional $(b2)$, $(b3)$, and $(b4)$ defined as follows: 

\begin{tabular}{lclclc} 
	& $(b2)$ &if $x\in R(u,v)$ and $y\in R(u,x)$, then $y\in R(u,v)$, & \\ 
	& $(b3)$ &if $x\in R(u,v)$ and $y\in R(u,x)$ then $x\in R(y,v)$, & \\
& $(b4)$ &if $x\in R(u,v)$ then $R(u,x)\cap R(x,v)= \{x\}$ & 
 \end{tabular}

The notation $x \in R(u,v)$ can be interpreted as $x$ is in between $u$ and $v$.   For example, the axiom  $(b2)$ can be interpreted as: if $x$ is between $u$ and $v$, and $y$ is between $u$ and $x$, then $y$ is between $u$ and $v$. Similarly we can describe all other axioms.  Hence we use the terminology of betweeness for an axiom on a transit function $R$. The above interpretation was the motivation for the concept of \emph{betweenness} in graphs using transit functions. It was formally introduced by Mulder in
\cite{muld-08} as those transit function that satisfy axioms $(b1)$ and $(b2)$. Here the axiom $(b1)$ is defined for every $u, v \in V$ and a transit function
$R$ as follows: \\ 
\begin{tabular}{lclclc} 
	& $(b1)$ & $x\in R(u, v), x\neq v\Rightarrow v\not\in R(u,x)$  
 \end{tabular}

The following implications can be easily verified for a transit function $R$ among axioms $(t1), (t2), (t3), (b1), (b3)$ and $(b4)$.
\begin{itemize}
\item  Axioms $(t1)$ and $(b4)$ implies axiom $(t3)$.
\item  Axioms $(t1), (t2), (t3)$ and $(b3)$ implies axiom $(b4)$ which implies axiom $(b1)$ ( that is , for a transit function $R, (b3)$ implies $(b4)$ implies $(b1)$)
\end{itemize} 

The converse of the above implications need not hold.  A transit function $R$ satisfying axioms $(b2)$ and $(b3)$ is known as a \emph{geometric transit function}.

The problem of characterizing the interval function of an arbitrary graph can be adopted for different graph classes; viz.,  characterizing the interval function of special graph classes using a set of first - order axioms on an arbitrary transit function. %But the first systematic study of the interval function is due to Mulder in \cite{mu-80}. 
Such a problem was first attempted by Sholander in \cite{Sholander-52} with a partial proof for characterizing the interval function of trees.  Chv\'{a}tal et al. \cite{Chvatal-11} obtained the completion of this proof.  Further new characterizations of the interval function of trees and block graphs are discussed in \cite{kaMc-01}.   Axiomatic characterization of the interval function of median graphs, modular graphs, geodetic graphs, (claw, paw)-free graphs and bipartite graphs are respectively described in \cite{chferna-16, mfn,mu-80,mune-09,nebe-95}.  

%The prime example of a transit function is the interval function $I$ of a connected graph $G$, given by
%\[I_G(u,v)=\{w \mid w \mbox{ lies on a shortest } u,v \mbox{-path} \}.\]
%When no confusion arises, we write $I$ instead of $I_G$. The first extensive study of the interval function is \cite{muld-80}, where the term interval function was coined.

%Axioms, as the above three, on a function $R$ on the set $V$ are called {\em transit axioms}.

In this paper, we continue the approach of characterizing the interval function of some related classes of graphs, namely, distance hereditary graphs, Ptolemaic graphs and  bridged graphs.  We fix the graph theoretical notations and terminology used in this paper.  Let $G$ be a graph and $H$ a subgraph of $G$. $H$ is called an \textit{isometric} subgraph of $G$ if the distance $d_H(u,v)$ between any pair of vertices, $u,v$ in $H$  coincides with that of the distance $d_G(u,v)$.  $H$ is called an induced subgraph if $u,v$ are vertices in $H$ such that $uv$ is an edge in $G$, then $uv$ must be an edge in $H$ also.  A graph $G$ is said to be $H$\emph{-free}, if $G$ has no induced subgraph isomorphic to $H$. Let $G_1, G_2, \ldots , G_k$ be graphs.  For a graph $G$, we say that $G$ is {\em $G_1, G_2, \ldots , G_k$-free} if $G$ has no induced subgraph isomorphic to $G_i$, $i=1, \ldots, k$.  % An induced cycle is a sequence of vertices $v_0, \ldots,v_{k-1},v_0$ such that $v_iv_j \in E$ if and only if $ \mid i-j\mid = 1({\rm mod\, }
%k)$. 
Chordal graph is an example of a graph $G$ which is defined with respect to an infinite number of forbidden induced subgraphs ($G$ is chordal if $G$ have no induced cycles $C_n$ with length $n$ more than three). There are several graphs that can be defined or characterized by a list of forbidden induced subgraphs or isometric subgraphs.  See the survey by Brandst\"{a}dt et al. \cite{brandstadt1999graph} and the information system \cite{isgci}, for such graph families.  
A graph $G = (V, E)$  is a \emph{bridged graph} if $G$ has no isometric cycles of length greater than $3$.  Clearly the family of bridged graphs contain the family of chordal graphs.  The graph $G$ is \emph{distance hereditary} if the distances in any connected induced subgraph $H$ of $G$ are the same as in $G$. Thus, any induced subgraph $H$ inherits the distances of $G$. The graph $G$ is a \emph{Ptolemaic graph} if $G$ is both chordal and distance hereditary. Both Ptolemaic graphs and distance hereditary graphs possess a characterization in terms of a list of forbidden induced subgraphs, while bridged graphs by definition itself possess an infinite list of forbidden isometric subgraphs. In this paper, our idea is to find suitable axioms that fail on every forbidden induced subgraph for the Ptolemaic and distance hereditary graphs, while that for the bridged graph is to find an axiom that fails on all of its forbidden isometric subgraphs, namely on all isometric cycles $C_n, n>3$.  

In addition to the geometric axioms $(b3)$ and $(b2)$, we consider the following betweenness axioms $(J0), (J2), (J2')$ and $(J3')$ for a transit function $R$ on $V$ for proving the characterizations of these classes of graphs.  
 
\begin{itemize}
	\item[$(J{0})$]:  For any pair of distinct vertices $u,v,x,y\in V$ we have $x\in R(u,y), y\in R(x,v)  \Rightarrow x\in R(u,v)$.
	\item[$(J{0'})$]:  $x\in R(u,y), y\in R(x,v), R(u,y)\cap R(x,v) \subseteq \{u,x,y,v\}  \Rightarrow x\in R(u,v)$.
%	\item[$(J1)$ ]:  $w\in R(u,v),w\not = u,v, \implies $ there exists $x\in R(u,w)\setminus R(v,w),y\in R(v,w)\setminus R(u,w)$, such that $R(x,w)=\{x,w\},R(y,w)=\{y,w\}$ and $w\in R(x,y)$
\item[$(J{2})$]:   $R(u,x) = \{u,x\}, R(x,v) = \{x,v\}, R(u,v) \neq  \{u,v\} \implies  x \in  R(u,v)$.
\item[$(J{2'})$]:  $x\in R(u,y), y\in R(x,v),  R(u,x) =\{u,x\}, R(x,y)=\{x,y\}, R(y,v) = \{y,v\}, R(u,v)\neq \{u,v\}\Rightarrow x\in R(u,v)$
%\item[$(J{3})$] :  $x\in R(u,y), y\in R(x,v), x\neq y, R(u,v)\neq \{u,v\}\Rightarrow x\in R(u,v)$. \\ 
\item[$(J{3'})$]:  $x\in R(u,y), y\in R(x,v), R(x,y)\neq \{x,y\}, R(u,v)\neq \{u,v\}\Rightarrow x\in R(u,v)$.
\end{itemize}
From the definition of the axioms, we observe the following. The axiom $(J2)$ is a simple betweenness axiom which is
always satisfied by the interval function $I$.  The axiom $(J2')$ is a natural extension of $(J2)$.  We provide examples in the respective sections for the independence of the axioms  $(J2')$ and $(J3')$ and $(J0')$. The axioms $(J2')$ and $(J3')$ were  first considered in \cite{mcjmhm-10} and later in \cite{Changat-22}. The axiom $(J0)$ first appeared in \cite{Sholander-52} for characterizing the interval function of trees. The axiom is discussed in \cite{Changat-22} for characterizing the interval function of a Ptolemaic graph $G$.  
We may observe  that both the family of bridged graphs $\mathcal{BG}$ and distance hereditary graphs $\mathcal{DH}$  is a strict super class of the family of Ptolemaic graphs,$\mathcal{PG}$, that is, $\mathcal{PG} \subsetneq  \mathcal{DH}$ and $\mathcal{PG} \subsetneq  \mathcal{BG}$ and  $\mathcal{BG}$  and $\mathcal{DH}$  coincide only in $\mathcal{PG}$. But, $\mathcal{BG} \nsubseteq  \mathcal{DH}$ and $\mathcal{DH} \nsubseteq \mathcal{BG}$,  This relation is also reflected in the implications between the axioms $(J0), (J2'), (J3')$ and $(J0')$.  From the definitions, we have the axiom $(J0)$ implies axioms $(J2')$ and  $(J3')$, and also $(J0)$ implies $(J0')$, while the reverse implications are not true.  In other words, axioms $(J2'), (J3')$ and $(J0')$ are weaker axioms than $J(0)$ and hence graphs whose interval function satisfy axioms  $(J2')$ and $(J3')$ will be a super class of graphs whose interval function satisfies $(J0)$.  Similarly graphs whose interval function satisfies axiom  $(J0')$  will be a super class of graphs whose interval function satisfies $(J0)$. See Figure~\ref{pgdhbg} for the relationships between the family $\mathcal{PG}$, $\mathcal{DH}$ and $\mathcal{BG}$.
\begin{figure}[H]
	\centering
	\includegraphics[width=10cm,height=4cm]{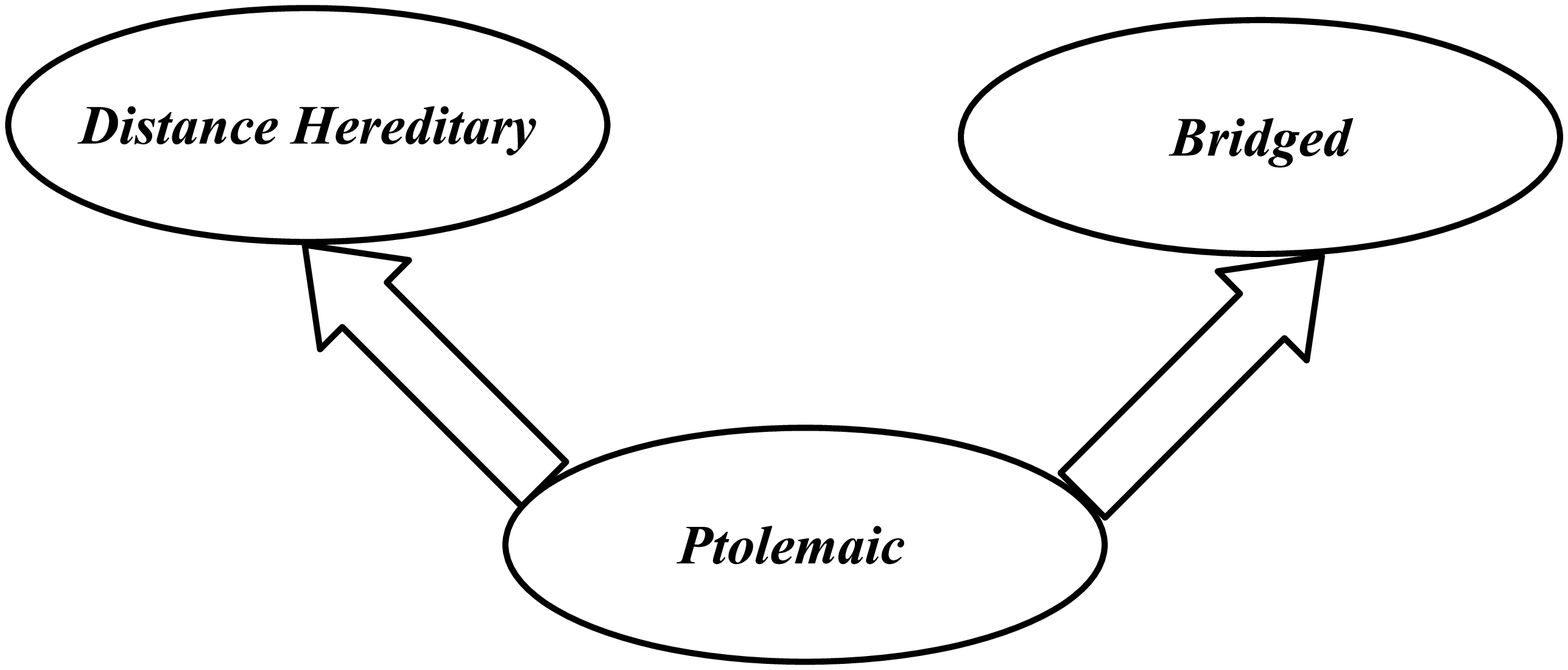}
	\caption{Relation between $\mathcal{PG}$, $\mathcal{DH}$ and $\mathcal{BG}$ }\label{pgdhbg}
\end{figure}

We organize the results as follows. In Section~\ref{DH}, we characterize the interval function of a distance hereditary graph, Ptolemaic graph in Section~\ref{Ptolemaic graphs},  bridged graph in Section~\ref{bridged} and in Section ~\ref{conclusion}, a discussion of the so called induced path transit function for the distance hereditary graphs respectively. 

%For this purpose we introduce the following axioms:

%\begin{itemize}
%\item[$(u1)$] if $R(u,v,w)\neq\emptyset$, then $R(u,v,w)\cap \{u,v,w\}=\{u\}$ or $\{v\}$ or $\{w\}$, for all $u,v,w\in V $.
%\end{itemize}

%So far we know that if the interval function $I$ of a connected graph $G$ satisfies axiom $(u1)$, then $G$ is claw-free. 

%\begin{itemize}
%\item[$(u2)$] if $R(u,v)\cap R(u,w)\setminus \{u\}\neq\emptyset$ and $R(v,w)\cap R(u,v)\setminus \{v\}\neq\emptyset$ and $R(u,w)\cap R(v,w)\setminus \{w\}\neq\emptyset$, then $R(u,v,w)\cap \{u,v,w\}=\{u\}$ or $\{v\}$ or $\{w\}$, for all $u,v,w\in V $.
%\end{itemize}

%So far we know that if the interval function $I$ of a connected graph $G$ satisfies axiom $(u2)$, then $G$ is (claw, net)-free. 

\section{Interval function of Distance hereditary graphs}\label{DH}

For the axiomatic characterization of the interval function of distance hereditary graphs,  we require the axioms $(J2')$ and $(J3')$. First we show that these axioms are independent with the Examples \ref{j2'nj3'} and \ref{j3'nj2'} below.  Also it is clear from the Figures \ref{hh3fan} and \ref{cnhdiam} that the axioms $(J2')$ and $(J3')$ are independent. The interval function $I_G$ doesn't satisfy axiom $(J2')$, while $I_G$ satisfy axiom $(J3')$ for the graphs in Figure~\ref{hh3fan}.  Also $I_G$ doesn't satisfy axiom $(J3')$, while $I_G$ satisfy axiom $(J2')$ for the graphs in Figure \ref{cnhdiam}.  By an $HHD3$ - fan - free graph $G$, we mean that $G$ is free from the House graph, the Hole graph (cycles $C_n$, $n\geq 5$), the Domino  graph and the $3$-fan (See the Figures~\ref{hh3fan} and \ref{cnhdiam} for these graphs). 

% \begin{enumerate}
 %	\item[$(J1)$ :]  $w\in R(u,v),w\not = u,v, \implies $ there exists $x\in R(u,w)\setminus R(v,w),y\in R(v,w)\setminus R(u,w)$, such that $R(x,w)=\{x,w\},R(y,w)=\{y,w\}$ and $w\in R(x,y)$
%\item[$ (J2')$ :] $x\in R(u,y), y\in R(x,v),  R(u,x) =\{u,x\}, R(x,y)=\{x,y\}, R(y,v) = \{y,v\}, R(u,v)\neq
 %\{u,v\}\Rightarrow x\in R(u,v)$\\
%\item[$(J3)$ : ] $x\in R(u,y), y\in R(x,v), x\neq y, R(u,v)\neq \{u,v\}\Rightarrow x\in R(u,v)$. \\ 
 %\item[$(J3')$ : ]  $x\in R(u,y), y\in R(x,v), R(x,y)\neq \{x,y\}, R(u,v)\neq \{u,v\}\Rightarrow x\in R(u,v)$.
 %\end{enumerate}
 
%It follows from the definition of the axioms that $(J3) \implies (J2')$ and $(J3)\implies (J3')$, but neither $(J2')$ implies $(J3')$ nor $(J3')$ implies $(J2')$.  

\begin{example}[$(J2')$ but not $(J3')$ ]\label{j2'nj3'}$~$\\
Let $V=\{u,v,w,x,y,z\}$.  Define a transit function $R$ on $V$ as follows.  $R(u,x)=\{u,x\}, R(u,w)=\{u,x,z,w\}=R(x,z), R(u,z)=\{u,z\}, R(u,v)=\{u,z,v\},R(x,y)=\{x,w,y\}, R(u,y)=V=R(x,v),R(x,w)=\{x,w\},R(z,w)=\{z,w\}, R(z,y)=\{z,w,y,v\}=R(w,v), R(z,v)=\{z,v\}, R(w,y)=\{w,y\}, R(y,v)=\{y,v\}$ and $R(x,x)=\{x\}$.  We can easily see that $R$ satisfies $(J2')$.  But we can see that $x\in R(u,y), y\in R(x,v), R(x,y)\neq \{x,y\}$ and $R(u,v)\neq \{u,v\}$ but $x\notin R(u,v)$.  So $R$ does not satisfy $(J3')$.
\end{example}
\begin{example}[$(J3')$ but not $(J2')$ ]\label{j3'nj2'}$~$\\
	Let $V=\{x,y,u,v,w\}$.  Define a transit function $R$ on $V$ as follows:$R(u,x)=\{u,x\}, R(u,y)=\{u,x,y\}, R(u,v)=\{u,w,v\}, R(u,w)=\{u,w\},  R(y,w)=\{y,x,v,w\}=R(x,v), R(y,v)=\{v,y\}, R(x,y)=\{x,y\},R(x,w)=\{x,w\}, R(v,w)=\{v,w\}$ and $R(x,x)=\{x\}$.  We can see that $R$ satisfies $(J3')$. But  $x\in R(u,y), y\in R(x,v), R(u,x)=\{u,x\}, R(x,y)=\{x,y\}, R(y,v)=\{y,v\}$, but $x\notin R(u,v)$.  Hence $R$ does not satisfy $(J2')$.
\end{example}
The following results are proved in \cite{Changat-22} and \cite{mcjmhm-10}.

\begin{proposition}\cite{Changat-22}\label{(j2')}
For every graph $G$, $I_G$ satisfies the $(J2')$ axiom if and only if $G$ is house, $C_5$, $3$-fan free. 
\end{proposition}

\begin{lemma}\cite{mcjmhm-10}\label{DHG}
 Let $R$ be a transit function on a non-empty finite set $V$ satisfying the axioms  $(b1), (J2), (J2')$ and $(J3')$ with underlying graph
$G_R$. Then $G_R$ is $HHD$ -free.
\end{lemma}

%\begin{theorem}\label{hhd}\cite{Changat-22}
%For every graph $G$, $I_G$ satisfies the $(J3)$ axiom if and only if $G$ is $HHP$, $3$ - fan free graph. 
%\end{theorem}

Bandelt and Mulder obtained a forbidden induced subgraph
characterization of distance hereditary graphs in \cite{bamu-86}. We
quote the theorem as

\begin{theorem}\cite{bamu-86}\label{HHD3-fan}
A graph $G$ is distance hereditary if and only if $G$ is $HHD3$ - fan-free.
\end{theorem}

We state a related result from \cite{Changat-22} using the axiom $(J3)$ defined for a transit function $R$ as $$'' x\in R(u,y), y\in R(x,v), x\neq y, R(u,v)\neq \{u,v\}\Rightarrow x\in R(u,v)'' $$
 Note that a $P$ graph is the graph obtained by adding a pendent edge on an induced 4-cycle, $C_4$.  It follows from the definition that axiom $(J3)$ implies both the axioms $(J2')$ and $(J3')$, but the reverse implications are not true.  We quote the result.
 
\begin{theorem}\label{hhd1}\cite{Changat-22}
For every graph $G$, $I_G$ satisfies the $(J3)$ axiom if and only if $G$ is $HHP3$ - fan - free graph. \end{theorem}

It may be observed that a $P$ graph is a distance hereditary graph and hence the class of $HHP3$ - fan - free graphs is a proper subclass of the class of $HHD3$ - fan - free graphs (distance hereditary graphs). The proof of the next theorem characterizing the class of distance hereditary graphs follows the same lines of ideas as in the proof of Theorem~\ref{hhd1} with modifications since the axioms $(J2')$ and $(J3')$ are weaker axioms than the axiom $(J3)$.  

%Using the forbidden subgraph characterization, we have the following theorem.

\begin{theorem}\label{J2J3Prime}
	Let $G$ be a connected graph. The interval function $I_G$ satisfy the axioms $(J2')$  and $(J3')$  if and only if $G$ is a distance hereditary graph.
\end{theorem}

\begin{proof}
	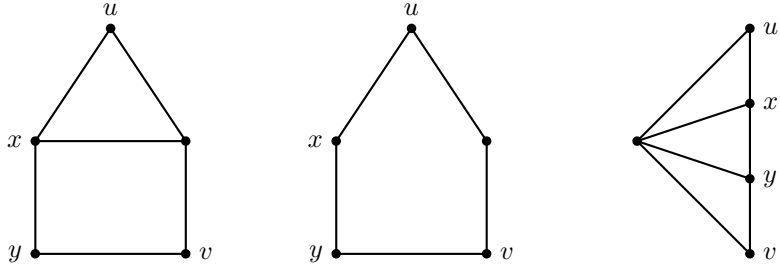
\begin{figure}[H]
		\centering
		\psscalebox{1}
		{	
			{
			%	\begin{pspicture}[showgrid](-2,0)(12,4)
				%	\begin{pspicture}(-2,0)(13,5)%(0,0)(10,5)
				\begin{pspicture}(-4,0)(9,3)
				%House
				\psdots(-2,0)(-2,1.5)(-1,3)(0,1.5)(0,0)%(4.5,2.5)
				\psline(-2,0)(-2,1.5)(-1,3)(0,1.5)(0,0)(-2,0)
				\psline(-2,1.5)(0,1.5)
				%	\psdots(3,1)(5,1)(7,1)%(4.5,2.5)
				%\psdots(4,3)(6,3)
				%	\psline(4,3)(5,1)(6,3)
				\uput[l](-2,0){$y$}
				\uput[l](-2,1.5){$x$}
				\uput[u](-1,3){$u$}
				\uput[r](0,0){$v$}
				%Hole
				\psdots(2,0)(2,1.5)(3,3)(4,1.5)(4,0)
				\psline(2,0)(2,1.5)(3,3)(4,1.5)(4,0)(2,0)
		%	%	\psline[linestyle=dotted](1,0)(3,0)
				\uput[l](2,0){$y$}
				\uput[l](2,1.5){$x$}
				 \uput[u](3,3){$u$}
				\uput[r](4,0){$v$}
				%Domino
%				\psdots(4,0)(4,1.5)(4,3)(6,3)(6,1.5)(6,0)
%				\psline(4,0)(4,3)(6,3)(6,0)(4,0)
%				\psline(4,1.5)(6,1.5)
%				\uput[l](4,0){$x$}
%				\uput[r](6,0){$u$}
%				\uput[l](4,1.5){$y$}
%				\uput[r](6,3){$v$}
				%3-fan
				\psdots(6,1.5)(7.5,0)(7.5,1)(7.5,2)(7.5,3)
				\psline(7.5,0)(7.5,3)
				\psline(6,1.5)(7.5,0)
				\psline(6,1.5)(7.5,1)
				\psline(6,1.5)(7.5,2)
				\psline(6,1.5)(7.5,3)
				\uput[r](7.5,0){$v$}
				\uput[r](7.5,1){$y$}
				\uput[r](7.5,2){$x$}
				\uput[r](7.5,3){$u$}
				%6 cycle and a Triangle
%				\psdots(10.5,0)(9.5,1)(9.5,2)(10.5,3)(11.5,2)(11.5,1)
%				\psline(10.5,0)(9.5,1)(9.5,2)(10.5,3)(11.5,2)(11.5,1)(10.5,0)
%				\psline(9.5,2)(11.5,2)
%				\uput[r](11.5,1){$v$}
%				\uput[u](10.5,3){$u$}
%				\uput[l](9.5,1){$x$}
%				\uput[d](10.5,0){$y$}
%				%6 cycle with two-path joining two  diametrical vertices
%				\psdots(4,-2)(4,-3.5)(4,-5)(6,-5)(6,-3.5)(6,-2)
%				\psline(4,0)(4,3)(6,3)(6,0)(4,0)
%				\psline(4,1.5)(6,1.5)
%				\uput[l](4,0){$x$}
%				\uput[r](6,0){$u$}
%				\uput[l](4,1.5){$y$}
%				\uput[r](6,3){$v$}
%				%	\uput[d](3,1){$v$} \uput[d](5,1){$y$}\uput[d](7,1){$u$}\uput[u](6,3){$x$}%\uput[u](3,2.5){$u_4$}
				%	\psline(3,1)(7,1)(6,3)(4,3)(3,1)
				%			\uput[d](0,0){$a$}\uput[r](0,2){$d$}\uput[u](0,4){$c$}\uput[l](-2,2){$b$}
				%			\psline(-2,2)(0,0)
				%			\psline(-2,2)(0,2)
				%			\psline(-2,2)(0,4)
				%			% A - graph		
				%			\psdots(5,3)(4,1)(6,1)(2.5,0)(7.5,0)
				%			\psline(2.5,0)(4,1)(6,1)(7.5,0)
				%			\psline(4,1)(5,3) 
				%			\psline(6,1)(5,3)			
				%			
				%			\uput[u](5,3){$a$}\uput[ur](6,1){$b$}\uput[ur](7.5,0){$c$}\uput[ul](2.5,0){$d$}
				%			%K_4-e, diamond
				%			\psdots(11,4)(9,2)(11,0)(13,2)
				%			\psline(11,4)(9,2)(11,0)(13,2)(11,4)
				%			\psline(9,2)(13,2)
				%			\uput[l](9,2){$a$}\uput[r](13,2){$b$}\uput[u](11,4){$c$}\uput[d](11,0){$d$}
				%			%				\psline(5,1)(5,3)(9,1)(9,3)(5,1)
				\end{pspicture}
			}
		}
		\caption{House, $C_5$, $3$ - fan}\label{hh3fan}
	\end{figure}
	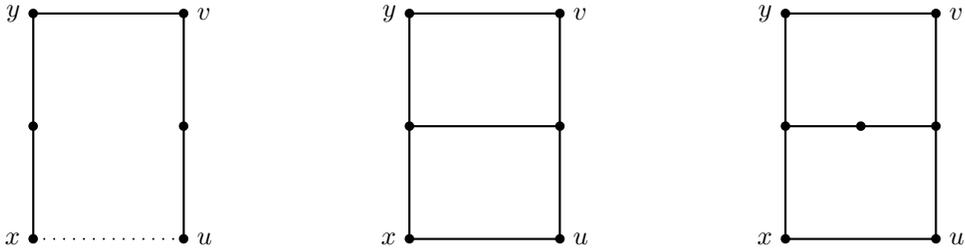
\begin{figure}[H]
		\centering
		\psscalebox{1}
		{	
			{
				%	\begin{pspicture}[showgrid](-2,0)(12,3)
				%	\begin{pspicture}(-2,0)(13,5)%(0,0)(10,5)
				\begin{pspicture}(-2,0)(11,3)
			%C_n,n>6
			\psdots(-1,0)(-1,1.5)(-1,3)(1,3)(1,1.5)(1,0)
				\psline(-1,0)(-1,3)(1,3)(1,0)%(4,0)
				\psline[linestyle=dotted](-1,0)(1,0)
					\uput[l](-1,0){$x$}
					\uput[l](-1,3){$y$}
					\uput[r](1,0){$u$}
					\uput[r](1,3){$v$}
					%Domino
						\psdots(4,0)(4,1.5)(4,3)(6,3)(6,1.5)(6,0)
						\psline(4,0)(4,3)(6,3)(6,0)(4,0)
						\psline(6,1.5)(4,1.5)
						\uput[l](4,0){$x$}
						\uput[r](6,0){$u$}
						\uput[l](4,3){$y$}
						\uput[r](6,3){$v$}
					%6 cycle with two-path joining a pair of vertices  diametrical vertices
				\psdots(9,0)(9,1.5)(9,3)(11,3)(11,1.5)(11,0)(10,1.5)%(5,2)
			\psline(9,0)(9,3)(11,3)(11,0)(9,0)
			\psline(9,1.5)(11,1.5)
			\uput[l](9,0){$x$}
			\uput[r](11,0){$u$}
			\uput[l](9,3){$y$}
			\uput[r](11,3){$v$}
			%\psline(4,3)(5,2)(6,0)
			%\psline(6,3)(5,1)(4,0)
%			\uput[l](4,3){$x$}
%			\uput[r](6,1.5){$u$}
%			\uput[l](4,1.5){$y$}
%			\uput[l](4,0){$v$}
				\end{pspicture}
			}
		}
	\caption{$C_n,n\ge 6$, Domino, $C_6$ cycle with a path joining two  diametrical vertices}\label{cnhdiam}
\end{figure}
We use the fact from Theorem~\ref{HHD3-fan} that distance hereditary graphs are precisely $HHD3$ - fan - free graphs for the proof. 
Suppose that the interval function $I_G$ of $G$ satisfy the axioms $(J2')$ and $(J3')$. To prove that $G$ is $HHD3$ - fan - free, assume the contrary that $G$ contains a house, a hole, or a domino  or a $3$-fan as an induced subgraph.  A graph with an induced house or $3$ - fan or a $C_5$ 
doesn't satisfy $(J2^\prime)$ ( The vertices $u,x,y,v$ in Figure~\ref{hh3fan} doesn't satisfy the axiom $(J2')$).
For an isometric hole $C_n$, $n\geq 6$, we choose vertices $u,x,y,v$ as shown in Figure
\ref{cnhdiam}, to prove that $(J3')$ is violated. If $G$ has a domino $D$, 
which is an isometric subgraph of $G$ with vertices as shown in Figure~\ref{cnhdiam}, then $x\notin I_G(u,v)$. If $D$ is not isometric, then there is
a vertex $z$ adjacent to both $u$ and $y$ or $v$ and $x$. In this case, the graph
induced by $u,a,x,y,z$ is either a $C_5$  or a house or a $3$-fan. \\
Let $C_n, n \geq 6$ be a hole in $G$ that is not isometric and assume that $n$ is minimum. Clearly $n\geq  6$ and there exist $u, v \in V(C_n)$
such that $d_G(u, v) < d_{C_n} (u, v)$. Let $P = p_0,\ldots,p_k$ be a $u, v$-geodesic; we may choose $P$ such that $k$ is minimum. Let $Q$ and
$R$ be the $u, v$-paths on $C_n$ where $Q = q_0,\ldots, q_\ell$, $R = r_0,\ldots,r_m$ and $u = p_0 = q_0 = r_0$ and $v = p_k = q_\ell = r_m$. Clearly
$2\leq k < \ell$  and $k < m$ and we may assume without loss of generality that $\ell\leq m$. Moreover, we can choose $P$ such that
the cycle $C$ induced by $V(P)\cup V(Q)$ has minimum length. In particular, by the choice of $C$, $p_{k-1}$ is not adjacent to any of
$\{q_{k+1},\ldots, q_{\ell-1}\}$ whenever $\ell > k + 1$. 
%If there are no edges between inner vertices of $P$ and $Q$ or $P$ and $R$, then we have an induced hole shorter than $n$, which is a contradiction with the minimality of $n$. 
By minimality of $k$, $p_i$, $i\in \{1,\ldots, k - 2\}$ can be
adjacent to $q_j$ (or $r_j$) only if $i = j$ or $j = i + 1$.
If $\ell\geq k + 3$, then $\{p_{k-1}, p_k, q_k, q_{k+1}, \ldots, q_{\ell-1} \}$ form (together with possibly some additional vertices of $P$ or $Q$) an
induced hole shorter than $n$, which is not possible. If $\ell = k + 2$, then $p_{k-1}$ and $q_k$ must be adjacent; otherwise we have a
shorter hole than $n$ on vertices $\{p_{k-1}, p_k, q_k, q_{k+1}\}$ together with $p_{k-2}$ or $q_{k-1}$ (and possibly some other vertices of $P$ or $Q$). But then
$A = \{p_{k-1}, p_k, q_k, q_{k+1} \}$ form a $4$-cycle and we have an induced domino  on $A\cup  q_{k-1}, p_{k-2}$ when $p_{k-1}q_{k-1}, p_{k-2}q_{k-2}\notin E(G)$ and an induced house on the same vertices when $p_{k-1}q_{k-1}\in  E(G)$.
Let now $\ell = k + 1$. First note that $p_1$ must be adjacent to at least one of $\{q_1, q_2\}$ and of $\{r_1, r_2\}$ by the minimality of $n$, since there are no isometric holes. Let first $k = 2$. If either $q_1$ or $q_2$ is not adjacent to $p_1$,
we have an induced house as a subgraph. Otherwise we have an induced $3$ - fan  on $\{p_1, q_0, q_1, q_2, q_3\}$.
Let now $k > 2$. By the above, $p_2$ is not adjacent to $q_1$ or to $r_1$ and $p_1$ is not adjacent to $q_3$ or $r_3$. Suppose that at least one of
$q_2$ or $r_2$, say $q_2$, is adjacent to $p_1$. If both $q_1$ and $r_1$ are adjacent to $p_1$, we have an induced $3$ - fan on $\{p_1, r_1, q_0, q_1, q_2\}$. If $q_1$ is
adjacent to $p_1$ but $ r_1$ is not, we have an induced house on $\{p_1, q_0, q_1, r_1, r_2\}$. If both $q_1$ and $r_1$ are not adjacent to $p_1$ , $r_2$  adjacent to  $p_1$ and $q_2$ adjacent to $p_1$, we have an induced domino on $\{p_1, q_0, q_1, q_2, r_1,r_2\}$. Hence $q_2p_1\notin  E(G)$ but $q_1p_1 \in E(G)$ and by a similar argument $r_2p_1\notin  E(G)$ but $r_1p_1 \in E(G)$. Since $q_2p_1, q_1p_2\notin  E(G)$, $p_2$ and $q_2$ must be adjacent since there are no isometric holes or by the minimality of $n$. This yields an induced house on $\{p_2, p_1, q_0, q_1, q_2\}$. Finally, If the induced $C_n$ is not isometric with $n=6$, the only case left is the one in Figure~\ref{cnhdiam}, and by choosing the vertices $u,x,y,v$ as in the figure, it follows that the axiom $(J3')$ is violated. Therefore, when $G$ has an induced house, hole, domino or $3$ - fan, either the axiom $(J2')$ or $(J3')$ is violated. 

Conversely  assume that axiom  $(J2')$ or $(J3')$ is not satisfied by the interval function $I_G$ of $G$. It is already known from Proposition~\ref {(j2')} that if axiom $(J2')$ is not satisfied, then $G$ has an induced $C_5$, House or $3$ - fan.  Now suppose axiom $(J3')$  is not
satisfied. Then there exists distinct vertices $u,x,y,v$ in $V$ such that
$x\in I_G(u,y), y\in I_G(x,v)$ , $I_G(x,y) \neq \{x,y\}$, $I_G(u,v)\neq \{u,v\}$ and $x\notin I_G(u,v)$.
Let $P$ be a $u,y$-geodesic containing $x$ and $Q$ be a $x,v$ - geodesic containing $y$.
%Since $d_P(x, y)=d_Q(x, y)$, we can replace the
%segment $x\rightarrow Q\rightarrow y$ of $Q$ by $x\rightarrow
%P\rightarrow y$.
We claim that $$u\rightarrow P\rightarrow x\rightarrow P \rightarrow
y\rightarrow Q\rightarrow v$$ is a $u,v$-path.

Since $d_P(x,y)=d_Q(x,y)$ we see that $ x\rightarrow P \rightarrow
y\rightarrow Q\rightarrow v$ is a $x,v$ geodesic. Therefore
$(x\rightarrow P\rightarrow y)\cap (y\rightarrow Q\rightarrow
v)=\{y\}$, for otherwise we may find a shorter path from $x$ to $v$.

Now we claim that $(u\rightarrow P\rightarrow x)\cap (y\rightarrow
Q\rightarrow v)=\emptyset$. Assume on the contrary, that $w\in
(u\rightarrow P\rightarrow x)\cap (y\rightarrow Q\rightarrow v)$,
and let $w$ be the last vertex of the intersection (when going from
$u$ to $x$ along $u\rightarrow P\rightarrow x$ path). Then
$d_P(w,y)=d_Q(y,w)$, and since $w\neq y$ we find that $d_P(w, x)<
d_Q(w, y)$.  Hence we have that the path $x\rightarrow P\rightarrow
w\rightarrow Q\rightarrow v$ is shorter  than $x\rightarrow
Q\rightarrow y\rightarrow w \rightarrow  Q\rightarrow v$, a
contradiction.

Hence $u\rightarrow P\rightarrow x\rightarrow P \rightarrow
y\rightarrow Q\rightarrow v$ is a $u,v$-path.  Since $x\notin
I_G(u,v)$, $u\rightarrow P\rightarrow x\rightarrow P\rightarrow
y\rightarrow Q\rightarrow v$ is not an $u,v$ geodesic. If $R$ is any
$u,v$-geodesic, then $x\notin V(R)$. Fix an $u,v$-geodesic $R$. Let
$a$ be the last vertex on $P$ before $x$ that is on $R$ and $b$ be
the first vertex on $Q$ after $y$ that is on $R$. Note that such
vertices always exists, since $u\in P\cap R$ and $v\in Q\cap R$. On
the other hand, note that $b$ can be equal to $y$, but $a\neq x$.
Label vertices of the path $a\rightarrow P\rightarrow y\rightarrow
Q\rightarrow b$ by $a=b_0,b_1,\ldots b_{\ell}=b$.  Label vertices of
the $a,b$-subpath of $R$ as $a=a_0,a_1, \ldots ,a_k=b$.  Clearly
$\ell \ge 4$ and $k\ge 2$ and $k<\ell $.  Path $b_0b_1\ldots
b_{\ell}$ is not necessarily an induced path. If not, we choose
among all chords $b_ib_j$ the one with maximal $j-i$ and replace the
part $b_i\ldots b_j$ by this chord. Vertices of this new path are
denoted by $a=c_0,c_1,\ldots ,c_t=a_k=b$, where still $t>k\geq 2$ by
the choice of $a$ and $b$. But $a_0a_1\ldots a_k$ is an induced
path, since it is a shortest path. Note that $c_1$ is not adjacent
to $a_2,\ldots ,a_k$ by the choice of $a$. Hence $c_1$ can be
adjacent only to $a_1$. Similarly $c_h$, for $h\geq 2$, cannot be
adjacent to $a_{h+1},\ldots ,a_k$. We consider the following two
cases.\\

\noindent 
\textbf{Case 1:} $a_1c_1\notin E(G)$.\\
If also $a_1c_2\notin E(G)$, we have an induced cycle of length $\ge
6$. So let $a_1c_2\in E(G)$. Also $a_2c_2\notin E(G)$ and
$a_1c_3\notin E(G)$, otherwise we have a house on vertices
$c_0,c_1,c_2,a_1$, and $a_2$ or $c_3$ respectively. This implies
that $c_3a_2$ is an edge and we have an induced domino
($c_1,c_2,c_3,a_2,a_1,a_0$), otherwise we have an induced cycle of length
$\geq 6$. \\

\noindent \textbf{Case 2:} $a_1c_1\in E(G)$.\\
If $c_3=a_2$, we have
a house if $a_1c_2\notin E(G)$ or a $3$ - fan on vertices
$a_0,a_1,c_1,c_2$, and $c_3$. Hence $c_3\neq a_2$ and also
$c_3a_1\notin E(G)$. We get an induced cycle of length $\ge 5$, if
$c_2$ is not adjacent to at least one of $a_1$ or $a_2$. If first
$c_2a_2\in E(G) $ and $c_2a_1\notin E(G)$, we get a house on
vertices $a_0,a_1,c_1,a_2$, and $c_2$. Let now $c_2a_2\notin E(G) $
and $c_2a_1\in E(G)$. Since $c_3\neq a_2$, there exists $c_4$. If
$c_4a_1\in E(G)$, we get a house on vertices $a_1,c_1,c_2,c_3$, and
$c_4$. Also $c_3a_2\notin E(G)$, since otherwise we get a house on
vertices $a_1,c_1,c_2,c_3$, and $a_2$. But now we have an induced
path $c_4c_3c_2a_1a_2$ which lead to an induced hole. Finally, if
$c_2a_2\in E(G) $ and $c_2a_1\in E(G)$, we get a $3$ - fan on vertices
$a_0,a_1,c_1,a_2$, and $c_2$. 
Thus in all cases, we get either an induced house, hole, domino, or $3$ - fan, and thus the proof is completed.
\end{proof}

 We need the following Lemma 
 
\begin{lemma}\label{HHD3fan}
 Let $R$ be a transit function on a non-empty finite set $V$ satisfying the axioms  $(b3), (J2), (J2')$ and $(J3')$ with underlying graph
$G_R$. Then $G_R$ is $HHD3$ - fan - free.
\end{lemma}

\begin{proof}
Since for a transit function axiom $(b3)$ implies axiom $(b1)$, by Lemma~\ref{DHG}, $G_R$ is $HHD$-free. We prove that $G_R$ is also $3$ - fan - free. Suppose on the contrary, $ G_R$ contains a $3$ - fan with vertices $u,v,x,y,z$ as an induced subgraph. Let $u,x,y,v$ be the path of length three and $z$ be the vertex adjacent to all of $u,v,x,y$.  Since $R(u,x)=\{u,x\},R(x,y)=\{x,y\}, R(y,v)=\{y,v\}$ and $R(u,y)\neq \{u,y\}, R(x,v)\neq\{x,v\}$, we have by axiom $(J2), x\in R(u,y)$, and $y\in R(x,v)$.  Again since  $R(u,v)\neq \{u,v\}$ by $(J2'), x\in R(u,v)$. Again, $R(x,z)=\{x,z\}$ and $R(z,v)=\{z,v\}$, $R(x,v)\neq \{x,v\}$, by axiom $(J2)$, we have $z\in R(x,v)$. Now, we have $x\in R(u,v)$ and $z\in R(x,v)$. Hence by axiom $(b3)$, we have $x\in R(u,z)$, which is a contradiction and hence the lemma is proved.
\end{proof}

\section{Axiomatic characterization of the interval function of Ptolemaic graphs }\label{Ptolemaic graphs}

For the axiomatic characterization of $I_G$ of a Ptolemaic graph $G$, the essential axiom is $(J0)$. Ptolemaic graphs are chordal graphs that are $3$ - fan  - free.  Changat et al. in \cite{Changat-22} characterized the graphs for which the interval function satisfies the axiom $(J0)$ as follows.

%\begin{itemize}
%\item[$(J{0})$]  For any pairwise distinct elements  $u, v, x, y \in V$ such that $x \in R(u, y)$ and $y \in R(x, v)$, we have $x \in R(u, v)$.
%\end{itemize}

\begin{theorem}\cite{Changat-22}\label{ijo}
Let $G$ be a graph. The interval function $I_G$ satisfies the axiom $(J0)$ if and only if $G$ is a Ptolemaic graph.
\end{theorem}

%\begin{itemize}
%	\item[$(J{2})$] : For a transit function $R$ on $V$, $R(u,x) = \{u,x\}, R(x,v) = \{x,v\},
   %   R(u,v) \neq  \{u,v\} \implies  x \in  R(u,v)$. 
%\end{itemize}

\begin{theorem}\label{j0j2}
	Let $R$ be any transit function defined on a non-empty set $V$.  If $R$ satisfies $(J0)$ and $(J2)$ then the underlying graph $G_R$ of $R$ is $C_n$-free for $n\ge 4$.
\end{theorem}
\begin{proof}
	Let $R$ be a transit function satisfying $(J0)$ and $(J2)$.  Let $G_R$ contains an induced cycle say $C_n=u_1u_2\ldots u_nu_1$.  Without loss of generality assume $C_n$ is the minimum such cycle (in the sense that length of the induced cycle is as small as possible). We prove for every $k\ge 4$.\\
	\noindent
	\textbf{Case: }$k=4$.\\
Now since $R(u_1,u_2)=\{u_1,u_2\}$ and $R(u_2,u_3)=\{u_2,u_3\}$, By $(J2)$ axiom we have $u_2\in R(u_1,u_3)$ in a similar fashion we can show that $u_3\in R(u_2,u_4)$.  Since $R$ satisfies (J0)-axiom we have $u_2\in R(u_1,u_4)$, which is a contradiction as $R(u_1,u_4)=\{u_1,u_4\}$.\\
\noindent
\textbf{Case:} $k=5$.\\
As in the above case, we can see that $u_4\in R(u_5, u_3)$, this together with  $u_3\in R(u_4,u_1)$ we have $u_4\in R(u_5,u_1)$, which is again a contradiction. \\
\noindent
Case: $k\ge 6$.\\
By repeated application of  $(J2)$-axiom, as in the above two cases, $u_{n-1}\in R(u_n,u_{n-2})$ with $u_{n-2}\in R(u_{n-1},u_1)$, by applying $(J2)$, we can see that $u_{n-1}\in R(u_n,u_1)$.  Here again a contradiction.\\

Hence, in all cases, we can see that $G_R$ does not contain $C_n$ as an induced subgraph  This completes theorem.
\end{proof}

%Now we define three axioms, which play a crucial role in the axiomatic characterization.

%\begin{tabular}{lclclc} 
%	& $(b1)$ & if $x\in R(u, v)$ and $x\neq v$, then $v\notin
%	R(u,x)$ & \\
%	& $(b2)$ &if $x\in R(u,v)$ and $y\in R(u,x)$,
%	then $y\in R(u,v)$ & \\
%	& $(b3)$ &if $x\in R(u,v)$ and $y\in R(u,x)$, then $x\in R(y,v)$ & 
 %\end{tabular}

%It can be seen easily that if a transit function $R$ satisfies axiom $(b3)$, then $R$ satisfies axiom $(b1)$ also.     
 The following straightforward Lemma for the connectedness of the underlying graph $G_R$ of a transit function $R$  is proved in \cite{mcjmhm-10}.

\begin{lemma}\label{connected}\cite{mcjmhm-10}
 If the transit function $R$ on a non-empty set $V$ satisfies axioms $(b1)$ and $(b2)$, then the underlying graph $G_R$ of $R$ is connected.
\end{lemma}

%Using the implication (J0) axiom implies (J0') axiom, we have the following corollary.
%\begin{corollary}
%Let $R$ be any transit function defined on the vertex set $V$  of a connected graph $G$.  If $R$ satisfies (J0') and (J2) then the underlying graph $G_R$ is $C_n$-free for $n\ge 4$.
%\end{corollary}
We have the following lemma.
\begin{lemma}\label{connect}
	If $R$ is a transit function on $V$ satisfying the axioms  $(J0)$ and $(b3)$, then $R$ satisfies axiom $(b2)$ and $G_R$ is connected.
\end{lemma}
\begin{proof}
Let $R$ satisfies axioms $(J0)$ and $(b3)$.  %Since $R$ satisfies axioms (J0), (J2), we see that $G_R$ is $C_n$-free for $n\ge 4$.    
To prove $R$ satisfies $(b2)$.  Since $R$ satisfies $(J0)$, For  $u,v,x,y\in V$, let  $x\in R(u,v)$,  and $y\in R(u,x)$.  Since $R$ satisfies $(b3)$, we have   $x\in R(u,v), y\in R(u,x)\implies x\in R(y,v)$.  Now  $y\in R(u,x), x\in R(y,v)$ and so  by axiom $(J0)$ , we have  $y\in R(u,v)$, which implies that  $R$ satisfies $(b2)$.  Connectedness of $G_R$ follows from Lemma~\ref{connected}, since $R$ satisfies axioms $(b1)$ and $(b2)$ as axiom $(b3)$ implies axiom $(b1)$. 
\end{proof}
%\begin{remark}
%	content...
%\end{remark}
\begin{example}[$(J0), (J2)$ and $(b2)$ but not  $(b3)$]$~$\\
Let $V=\{u,v,w,x,y\}$.  Let $R:V\times V\rightarrow 2^V$ be defined as follows.  $R(u,v)=V, R(u,x)=\{u,y,w,x\}, R(w,v)=\{x,w,y,v\}$ and in all other cases $R(a,b)=\{a,b\}$.  Since $G_R$ is a $3$ - fan, $R$ satisfies $(J0)$ and $(J2)$.  Next to show $R$ satisfies $(b2)$ axiom.  Since $R(u,v)=V$, we can see that $R(u,x)\subseteq R(u,v)$ for all $x\in R(u,v)$  so that for this pair $R$ satisfies $(b2)$ axiom.  Now consider $R(u,x)$, we can see that $a(\neq u,x)\in R(u,x)$ we have $R(u,a)=\{u,a\}$ and $R(x,a)=\{x,a\}$ which implies that $R$ satisfies $(b2)$ axiom for this pair too.  The case is similar for $R(w,v)$.  All other pairs corresponds to edges.  Hence we can see that $R$ satisfies $(b2)$-axiom.
Now $x\in R(u,v), y\in R(u,x)$ but $x\notin R(y,v)=\{y,v\}$, and $R$ violates $(b3)$  axiom.
\end{example}

\begin{theorem}\label{Ptolemaic}
Let $R$ be any transit function satisfying the axioms  $(b3), (J0)$ and $(J2)$ then $G_R$ is Ptolemaic and $R(u,v)=I(u,v)$.
\end{theorem}
\begin{proof}
Since $R$ satisfies the axioms  $(b3), (J0)$ and $(J2)$, we have that $G_R$ is a chordal graph by Theorem~\ref{j0j2}. To prove that $G_R$ is Ptolemaic, we have to show that $G_R$ is $3$ - fan - free. Suppose that $G_R$ contains an induced $3$ - fan with vertices $u,x,y,v,z$ with $u,x,y,v$ forming a path $P_4$ and $z$ as the vertex adjacent to all the vertices $u,x,y,v$.  Since $ux$ and $xy$ are adjacent and $uy$ is not an edge, by $(J2)$,  $x\in R(u,y)$.  Similarly $y\in R(x,v)$. Since $R$ is a transit function, by $(t2)$,  $y\in R(v,x)$ and $x\in R(y,u)$ and hence by $(J0)$,  $y\in R(u,v)$. Again, since $uz$ and $zy$ are edges and $uy$ is not an edge,  $z\in R(u,y)$. That is, $y\in R(u,v)$ and $z\in R(u,y)$, by $(b3)$, we have $y\in R(z,v)$, which is not true as $zv$ is an edge. That is, we have proved that $G_R$ is a chordal graph which is $3$ - fan  - free and hence $G_R$ is a Ptolemaic graph. By Lemma~\ref{connect}, $R$ satisfies axiom $(b2)$ and $(b1)$ and $G_R$ is connected.

Now we prove that  $R(u,v)=I(u,v)$ for all $u,v\in V$.  We prove by induction on the distance between $u$ and $v$.\\
\noindent
\textbf{Case when $d(u,v)=2$}.\\
Let $x\in I(u,v)$  Hence we can see that $ux,xv\in E$.   That is, $R(u,x)=\{u,x\}, R(x,v)=\{x,v\}$ and $R(u,v)=\{u,v\}$, since $R$ satisfies $(J2), x\in R(u,v)$.  Therefore $I(u,v)\subseteq R(u,v)$.  Conversely suppose $x\in R(u,v)$.  Suppose  $x\notin I(u,v)$.  Since $d(u,v)=2$  there exists at least one element $y\in I(u,v)$ such that $uy,yv$ are edges in $G_R$.  By assumption, $x$ is not adjacent to both $u$ and $v$.  Assume that $xu$ is not an edge.  Since $x\in R(u,v)$ and $R$ satisfies $(b2)$ and $(b1), R(u,x) \subset R(u,v)$ with $|R(u,x)|< |R(u,v)|$.  By applying axioms $(b2)$ and $(b1)$ continuously on $R(u,x)$, we get  vertices $x_i,x_{i+1},\ldots, x_k,x_{k+1}=x\in R(u,x)$  such that $R(x_{i},u) \subset R(x_{i+1},u)$ and $|R(x_{i},u)| < |R(x_{i+1},u)|$, for  $i= 1,\ldots,k$ and since $V$ is finite, $R(u,x_i)= \{u,x_i\}$, for some $i$, say $i=1$. That is, we have vertices $x_1,x_2,\ldots, x_k,x_{k+1}=x\in R(u,x)$ with $R(x_1,u)=\{x_1,u\}$. Let us assume that $R(x_1,y)\neq \{x_1,y\}$.  That is $x_1y\notin E(G_R)$. Consider vertices $x_1,u,y,v$. By $(J2), u\in R(x_1,y)$ and since $y\in R(u,v)$, by $(J0), u\in R(x_1,v)$. That is, $x_1\in R(v,u), u\in R(v,x_1)$ and hence by $(b3),  x_1\in R(u,u)$, a contradiction. Therefore $R(x_1,y)=\{x_1,y\}$. This implies that $y\in R(x_1,v)$ by axiom $(J2)$, provided $R(x_1,v)\neq \{x_1,v\}$, which implies that $x_1\in R(y,u)$ by $(b3)$, a contradiction since $R(u,y)=\{u,y\}$. Therefore $R(x_1,v)=\{x_1,v\}$. That is, we have $x\in R(u,v), x_1\in R(u,x)$ and hence by $(b3), x\in R(x_1,v)$, a final contradiction. Therefore $R(u,x)=\{u,x\}$. Similarly, we can prove that $R(v,x)=\{v,x\}$.  $x\in I(u,v)$ and hence $R(u,v) \subset I(u,v)$, which completes the proof when $d(u,v)=2$. 

%Hence $u\rightarrow v\rightarrow x\rightarrow u$ is a cycle $C$ of length at least $5$.  Without loss of generality, assume that the cycle $C$ is the smallest cycle.  Since $G_R$ does not contain any induced cycles of length $\ge 4$ and since $C$ is a cycle of smallest length, we  have that the vertex $y$ should be adjacent to both the neighbours $x'$ and $x''$ of $u$ and $v$ respectively, Otherwise we get that the  induced cycle $uy\rightarrow z\rightarrow xu$ is of length at least $4$.  In a similar fashion we can prove that $y$ is adjacent to the neighbours of $x$ and $x'$ (other than $u$ and $v$) if there are neighbours other than $u$ and $v$. In this case we see that the subgraph induced by $uyv\rightarrow xu$ will contain a $3$-fan as an induced subgraph, a contradiction since $G_R$ is Ptolemaic.  

Let us assume that the result holds for all distances less than $k>2$ and let $u$,$v$ be two vertices such that $d(u,v)=k>2$.  We first prove $I(u,v)\subseteq R(u,v)$.  Let $x\in I(u,v)$.   Since $d(u,v)>2$, we can find another vertex $y$ in the shortest $u,v$-path containing $x$.   Now since $I$ satisfies $(b2), I(u,x)\subseteq I(u,v), I(x,v)\subseteq I(u,v)$.  So by induction we have $I(u,x)=R(u,x)$ and $I(x,v)=R(x,v)$.  Also by $(b3)$ axiom $x\in I(u,y)=R(u,y)$, $y\in I(x,v)=R(x,v)$.  Then by $(J0)$ axiom $x\in R(u,v)$.  Hence $I(u,v)\subseteq R(u,v)$.
Let $x\in R(u,v)$.  If possible let $x\notin I(u,v)$.  Since $x\in R(u,v)$, by applying axioms $(b1)$ and $(b2)$ similarly as in the case of $d(u,v)=2$, we get vertices $x_1,x_2,\ldots, x_k,x_{k+1}=x\in R(u,x)$ with $R(x_1,u)=\{x_1,u\}$  such that $R(x_{i},u) \subset R(x_{i+1},u)$ and $|R(x_{i},u)| < |R(x_{i+1},u)|$, for  $i= 1,\ldots ,k$ and  $R(x_1,u)=\{x_1,u\}$. 
Let $y$ be a vertex such that $R(u,y)=\{u,y\}$ and $y\in I_{G_R}(u,v)$. Similar to the case of $d(u,v)=2$, we can prove that $R(x_1,y)=\{x_1,y\}$. That is $u,x_1,y$ form a $c_3$ in $G_R$.  Here there are two possibilities for $d(x_1,v)$. \\
\noindent
\textbf{Case (i): $d(x_1,v)= k$}. 
In this case, since $d(u,v)=k$ and $y$ is on a shortest $u,v$-path in $G_R$ with $d(y,v)= k-1$, we have that $y$ is on a shortest $x_1,v$-path in $G_R$, that is, $y\in I_{G_R}(x_1,v) \subseteq R(x_1,v)$. Therefore, we have $x_1\in R(u,v), y\in R(x_1,v)$ and hence by $(b3), x_1\in R(y,u)$, a contradiction as $R(y,u)=\{y,u\}$.\\
\noindent
\textbf{Case(ii): $d(x_1,v)= k-1$}. 
In this case, $x_1\in I_{G_R}(u,v)$. Since $x\in R(u,v)$ and so by $(b2)$ axiom, $R(x,v) \subseteq R(u,v)$. We have also $x\in R(u,v), x_1\in R(u,x)$ and hence by axiom $(b3)$, we have $x\in R(x_1,v)= I_{G_R}(x_1,v)$,  by induction hypothesis. That is $x\in I_{G_R}(x_1,v) \subseteq I_{G_R}(u,v)$, since $x_1\in R(u,v)$, which is a contradiction to our assumption. 
%Case(iii): $d(x_1,v)= k+1$. In this case, we have $u\in I_{G_R}(x_1,v)\subseteq R(x_1,v)$.  Now, we have $x_1\in R(u,v), u\in R(x_1,v)$ and hence by axiom (b3), $x_1\in R(u,u)$, a contradiction. \\
Therefore in all cases, we get contradictions to the assumption and hence our assumption is wrong, that is $x\in R(u,v)\subseteq I_{G_R}(u,v)$ and hence the theorem. 
\end{proof}
The following examples show that the axioms $(J0), (J2)$ and $(b3)$ are independent.

\begin{example}[$(J0), (J2)$ but not $(b3)$]$~$\\
	Let $V=\{a,b,c,d,e\}$ and define a transit function $R$ on $V$ as follows: $R(a,b)=\{a,b\}, R(a,c)=\{a,c\}, R(a,d)=\{a,b,c,d\},R(a,e)=V, R(b,c)=\{b,c\}, R(b,d)=\{b,d\}, R(b,e)=\{b,e\}, R(c,d)=\{c,d\}, R(c,e)=\{b,c,d,e\}, R(d,e)=\{d,e\}$.  We can see that $R$ satisfies $(J0)$ and $(J2)$.  But $d\in R(a,e), b\in R(a,d)$, but $d\notin R(b,e)$.  Therefore $R$ doesnot satisfy the $(b3)$ axiom.
\end{example}
\begin{example}[$(J2), (b3)$ but not $(J0)$]$~$\\
		Let $V=\{a,b,c,d,e\}$ and define a transit function $R$ on $V$ as follows: $R(a,b)=\{a,b\}, R(a,c)=\{a,c\}, R(a,d)=\{a,b,c,d\},R(a,e)=\{a,b,e\}, R(b,c)=\{b,c\}, R(b,d)=\{b,d\}, R(b,e)=\{b,e\}, R(c,d)=\{c,d\}, R(c,e)=\{b,c,d,e\}, R(d,e)=\{d,e\}$. Here $R$ Satisfies $(J2)$ and $(b3)$.  We can see that $c\in R(a,d), d\in R(c,e)$ but $c\notin R(a,e)$.  So $R$ doesnot satisfy $(J0)$.
\end{example}
\begin{example}[$(J0), (b3)$ but not $(J2)$]$~$\\
	Let $V=\{a,b,c,d,e\}$ and define a transit function $R$ on $V$ as follows: $R(a,e)=\{a,e\}, R(b,e)=\{b,e\}, R(a,b)=\{a,b,c\}$ and for all other pair $R(x,y)=\{x,y\}$ we can see that $R$ satisfies $(J0), (b3)$ .  But  since $e\notin R(a,b)$ we can see that  $R$ fails to satisfy $(J2)$.
\end{example}
From Theorem~\ref{Ptolemaic} and Theorem~\ref{ijo}, we have the following theorem characterizing the interval function of Ptolemaic graphs.

\begin{theorem}\label{Ptolemaic_Ch}
Let $R$ be a transit function on the vertex set $V$ of a connected graph $G$. Then $G$ is a Ptolemaic graph and $R$ coincides the interval function $I_G$ of $G$ if and only if $R$  satisfies the axioms  $(b3), (J0)$ and $(J2)$ and the axiom $R(u,v) = \{u,v\}$ implies that $uv\in E(G)$. 
\end{theorem}

\section{Interval function of bridged graphs}\label{bridged}

%\subsection{Modification of axiom $(J0')$ }

%We slightly modify the $(J0)$ axiom: $(J0') :x\in R(u,y), y\in R(x,v)$ and $(R(u,y) \cap (R(x,v) \subseteq \{ u,x,y,v\} \implies x\in R(u,v)$.   We have a straight forward  proposition. 
%\begin{example}[$(J0')\centernot \implies (J0)$]\label{j0'j0}
%\end{example}
From the definitions of  $(J0)$ and $(J0')$ it follows that $(J0)\implies (J0')$.  The  example~\ref{j0'nj0} shows that  $(J0')\centernot \implies (J0)$.
%\begin{proposition}[$(J0')\centernot \implies (J0)$]\label{j0'j0}
%The axiom $(J0)\implies (J0')$, but not conversely.
%\end{proposition}

\begin{example}[$(J0')\centernot \implies (J0)$]\label{j0'nj0}$~$\\
%From the definitions, it follows that  $(J0)\implies (J0')$. For proving that  $(J0')\centernot \implies (J0)$, we have the following example. 
Let $V=\{a,b,c,d,e\}$  Let $R:V\times V\rightarrow 2^V$ defined as follows. $R(a,e)=\{a,e\}, R(a,b)=\{a,b\}, R(b,e)=\{b,e\}, R(b,c)=\{b,c\}, R(c,e)=\{c,e\},R(c,d)=\{c,d\}, R(d,e)=\{d,e\}, R(a,c)=\{a,b,c,e\}, R(a,d)=\{a,e,d\},R(b,d)=\{b,c,d,e\}$.  We can see that $b\in R(a,c)$ and  $c\in R(b,d)$ but $b\notin R(a,d)$, so that $R$ doesnot satisfy $(J0)$ axiom.  We can see that there exists no $u,v,x,y$ and $z$ satisfying the assumptions of the axiom $(J0')$ and hence the axiom $(J0')$  follows trivially. 

\end{example}

We now prove the theorem characterizing interval function of bridged graphs. 

\begin{theorem}
	Let $G$ be a connected graph. The interval function $I_G$ satisfies the axiom $(J0')$ if and only if $G$ is a bridged graph.
\end{theorem}
\begin{proof}
Let  $G$ has an isometric cycle $C_k= v_1v_2\ldots v_k$, $k > 3$.  If $k$ is odd, say $k=2t+1, t\geq 2$, let $u=v_1, x= v_{t}, y= v_{t+1}, v= v_{t+2}$.  Then it is easy to see that $x\in I_G(u,y)$, $y\in I_G(x,v)$ and $I_G(u,y) \cap I_G(x,v) = \{x,y\}  \subseteq \{ u,x,y,v\}$.  If $k$ is even, say, $k=2t$, $t\geq 2$, let $u=v_1, x= v_{t}$, $y= v_{t+1}, v= v_{t+2}$.  Then it is easy to see that $x\in I_G(u,y), y\in I_G(x,v)$ and $I_G(u,y) \cap I_G(x,v) = \{x,y,v\} \subseteq \{ u,x,y,v\}$. In both cases of $k$ being odd or even, $x$ is not on any shortest $u,v$-path and hence $x\notin I_G(u,v)$. This implies that If $G$ has an isometric cycle, then $I_G$ doesn't satisfy the axiom $(J0')$, that is, we have proved that if $I_G$ satisfies axiom $(J0')$ implies that $G$ is bridged graph. 

Conversely, if $G$ is a bridged graph, then we claim that  $I_G$ satisfy the axiom $(J0')$. 
Suppose not. Then there exist vertices $u,x,y,v$ in $G$ satisfying the following.   A $u,y$-geodesic $P$  containing $x$, an $x,v$-geodesic $Q$ containing $y$  with $I_G(u,y)\cap I_G(y,v)\subseteq \{u,x,y,v\}$ such that $x$ is not on any $u,v$-geodesic in $G$. Then $x$ and $ y $ should be adjacent, since $I_G(u,y)\cap I_G(y,v)\subseteq \{u,x,y,v\}$

Using the same arguments as in the proof of Theorem~\ref{J2J3Prime}, we derive the following.

\begin{itemize}
	\item[i]:  $u\rightarrow P\rightarrow x\rightarrow P \rightarrow y\rightarrow Q\rightarrow v$ is a $u,v$-path, say $P'$.
	\item[ii]:  The last vertex on $P$ before $x$ that is on a shortest $u,v$-path $R$ containing $x$ is $a$
	and the first vertex on $Q$ after $y$ that is on $R$ is $b$.
	\item[iii]: An $a,b$-subpath of $R$, $R_{a,b}:  a= z_0z_1\ldots z_{t-1}z_t=b$,$(t\geq 1)$ and an $a, b$ induced path $P'': a=u_0u_1\ldots u_\ell=xu_{\ell+1} =y\ldots u_{\ell+s} =b, (  \ell+s \geq 2)$ containing $x$ and $y$, which is a subpath of of $P'$.
	\item[iv]:  The cycle $C$ formed by the vertices of $R_{a,b} \cup P''$  has length, $l(C)$, at least four.  
	
\end{itemize}
Now, $l(C)$ cannot be four, since $C$ is isometric, a contradiction to $G$ being a bridged graph, which implies that the length of the path  $R_{a,b}$, namely $t$ is strictly greater than one.
Since $G$ is a bridged graph, it follows that, there are chords  from vertices $z_i$, $i=1,\ldots t-1$ to vertices $u_1,\ldots,u_\ell=x,u_{\ell+1} =y, \ldots u_{\ell+s-1} $   of $P''$ so that the only isometric cycles are triangles.\\ 
\noindent
\textbf{Case 1: $ y \neq b $}.\\
We claim that the index $ t=\ell+s-1$ . Since  $P''$  is not a $u,v$-geodesic   containing $x$,
we have $ t \leq \ell+s-1$. If  $ t \leq \ell+s-2$, since the only isometric cycles are triangles, we get  $P$ is not a $u,y$-geodesic containing $x$ or $Q$ is not a  $x,v$-geodesic  containing $y$. There for  $ t=\ell+s-1$.  Also if  $1\leq i \leq \ell-1$ ,$1\leq j \leq \ell$, $z_i$  cannot be adjacent to $u_j$  for $ j\geq i+2 $ , since   $P$ is a $u,y$-geodesic containing $x$, and if $\ell+1\leq i \leq t-1,\ell\leq  j \leq \ell+s-1$, $z_i$  cannot be adjacent to $u_j$  for $ j\leq i-1 $ , since $Q$ is a  $x,v$-geodesic  containing $y$. Which implies that  $ z_\ell $ is adjacent to both  $u_\ell=x$ and $u_{\ell+1} =y$, otherwise there exist an induced $4$ - cycle on $ \{z_{\ell-1}, z_\ell, x,y \}  $ or $ \{ z_\ell, x,y, z_{\ell+1}\} $.  Then  the path $u\ldots az_1\ldots z_{\ell}  y$  is also a  $u,y$ shortest  path and  the path  $xz_{\ell} \ldots z_t=b\ldots v$ is also a  $x,v$-shortest path. This implies that the vertex $z_{\ell}$, which is different from $u,x,y,v$ also belongs to $I_G (u,y)\cap I_G(x,v)$, a contradiction to the hypothesis of the axiom $(J0')$.\\
\noindent
\textbf{Case 2: $ y=b $}. \\
Since $P$ is  a $u,y$-geodesic containing $x$, $ t=\ell+s $. In this case $ I_G(x,v) $ does not contain any of $ z_i, i=1,2,\ldots,t-1 $. Which implies that $I_G (u,y)\cap I_G(x,v)  \subseteq \{u,x,y,v\} $, a contradiction to the hypothesis of the axiom $(J0')$.
Therefore in both case $I_G$ satisfies the axiom $(J0')$, which completes the proof of  sufficiency part.   
\end{proof}

\section{Concluding Remarks}\label{conclusion}
We conclude the paper by discussing another graph transit function, namely the induced path transit function for a distance hereditary graph. 
%A $u,v$-path: $u,u_1,\ldots u_i, \ldots u_{k-1},v$ in $G$ in which there are no edges of the form $u_i,u_j$, where $j\neq i-1, i+1$ ( such an edge $u_iu_j$ with $|i-j|>1$ is a chord in $G$) is called an induced path or chordless path or monophonic path or minimal path. 
By replacing shortest paths by induced paths in a graph $G$, we get the induced path transit function $J_G$. This function is also well studied. For example, see the references;  \cite{chma-04, chma-06, momu-02, vel-93}.  Nebesk\'{y} proved in \cite{nebe-02} a very interesting result: there does not exist a characterization of the induced path function $J$ of a connected graph using a set of first - order axioms.  Changat et al. in \cite{mcjmhm-10} characterized the induced path transit function axiomatically on $HHD$-free and $HHP$-free graphs. 

Formally the induced path function $J_G$ of $G$, is defined as 
\[J(u,v)= \{w \mid w \mbox{ lies on an induced } u,v \mbox{-path} \}.\]

Next, we define an axiom $(J1)$, which is used in the following discussions. 

\begin{enumerate}
\item[$(J1)$:] $w\in R(u,v), w\neq u, v \Rightarrow $ there exists $u_1\in R(u,w)\setminus R(v,w), v_1\in R(v,w)\setminus R(u,w)$, such that
$R(u_1,w)=\{u_1,w\}$, $R(v_1,w)=\{v_1,w\}$ and $w\in R(u_1,v_1)$.
\end{enumerate}

The following result is proved in \cite{mcjmhm-10}.

\begin{theorem}\cite{mcjmhm-10}\label{induced}
 Let $R$ be a transit function on a non-empty finite set $V$ satisfying the axioms  $(b1), (b2), (J1), (J2), (J2')$ and $(J3')$ with underlying graph
$G_R$. Then $G_R$ is $HHD$-free and $R$ is precisely the induced path transit function of $G_R$. 
\end{theorem}

We have the following proposition for a transit function $R$ on $V$.

\begin{proposition}\label{b2b3j1}
If a transit function $R$ on $V$ satisfies the axioms $(b2)$ and $(b3)$ then $R$ satisfies axiom  $(J1)$.
\end{proposition}
\begin{proof}
	Let $V$ be any non-empty set, and $R$ be a transit function defined on $V$.  We know that $R$ satisfies $(b3)$ implies $R$ satisfies $(b1)$.  Let $w\in R(u,v)$.  Since $R$ satisfies $(b2)$, we can see that $R(u,w)\subseteq R(u,v)$.  Again since $R$ satisfies $(b3)$, we can see that $v\notin R(u,w)$ (other wise $R$ will not satisfy the $(b1)$ axiom).  So we have $R(u,w)\subsetneq R(u,v)$ and $R(v,w)\subsetneq R(u,v)$.  \\
	\textbf{Claim :} $R(u,w)\nsubseteq R(v,w)$ and $R(v,w)\nsubseteq R(u,w)$.\\
	If  $R(u,w)\subseteq R(v,w)$, then we can see that $u\in R(v,w)$ a contradiction to the fact that $R$ satisfies $(b1)$ axiom.  In a similar fashion if $R(v,w)\subseteq R(u,w)$, we will get a contradiction.\\
	So there exists a vertices $x_1\in R(u,w)\setminus R(v,w)$ and $y_1\in R(v,w)\setminus R(u,w)$.  Consider $R(x_1,w)$ and $R(y_1,w)$.  Since $R$ satisfies $(b2), (b1)$, we have (as in the above lines) $R(x_1,w)\subsetneq R(w,u)\subsetneq R(u,v)$ and $R(y_1,w)\subsetneq R(w,v)\subsetneq R(u,v)$.  Continuing like this we can get a sequence of vertices $x_1,x_2,\ldots, x_\ell$ and $y_1,y_2,\ldots, y_m$ so that $R(x_\ell,w)\subsetneq R(x_{\ell-1},w)\subsetneq\ldots\subsetneq R(x_1,w)\subsetneq R(w,u)\subsetneq R(u,v)$ and $R(y_m,w)\subsetneq R(y_{m-1},w)\subsetneq\ldots\subsetneq R(y_1,w)\subsetneq R(w,v)\subsetneq R(u,v)$, with $R(x_\ell,w)=\{x_\ell,w\}, R(y_m,w)=\{y_m,w\}$.  \\ 
	Without loss of generality we assume $x_\ell=x$ and $y_m=y$.
	Now since $R$ satisfies $(b3)$ we have  $w\in R(u,v), x\in R(u,w)\implies w\in R(x,v)$.  Again using the $(b3)$ axiom, $w\in R(x,v), y\in R(w,v)\implies w\in R(x,y)$.  Hence $R$ satisfies the $(J1)$ axiom.
\end{proof}
\begin{proposition}
If a transit function $R$ on $V$ satisfies the axioms 	$(J1), (b2)$ then $R$ satisfies $(b1)$.
\end{proposition}
\begin{proof}
	Let $V$ be any non empty set.  Let $R$ be any transit function defined on $V$.  Let $R$ satisfy the axioms $(J1), (b2)$.  If possible assume that $R$ doesnot satisfy $(b1)$.  Therefore there exists $u,v,w$ with $w\in R(u,v)$ and $v\in R(u,w)$.  Since $R$ satisfies $(b2)$, and $v\in R(u,w)$ we have $R(u,v)\subseteq R(u,w)$.  Again since $w\in R(u,v)$, we must have $R(u,w)\subseteq R(u,v)$.  So we have $R(u,w)=R(u,v)$.  Now since $R$ satisfies $(J1)$, there should exist an element $y\in R(v,w)\setminus R(u,w)$.  Now we have $R(v,w)\subseteq R(u,v)$ so we have $R(v,w)\setminus R(u,w)=R(v,w)\setminus R(u,v)=\emptyset$, a contradiction to the assumption of $(J1)$.  So our assumption is wrong and $R$ satisfies $(b1)$.
\end{proof}

The example below shows that axioms $(b2), (J1)$, doesn't imply axiom $(b3)$. 

\begin{example}[$(b2), (J1)$ but not $(b3)$]$~$\\
	Let $V=\{u,v,x,y,z\}$.  Define $R$ on $V$ as follows.$R(u,v)=V, R(u,y)=\{u,y\}, R(u,x)=\{u,y,z,x\}, R(u,z)=\{u,z\}, R(z,y)=\{z,y\}, R(z,x)=\{z,x\}$, $R(z,v)=\{z,y,x,v\}, R(x,v)=\{x,v\}, R(x,y)=\{x,y\}, R(y,v)=\{y,v\}$.  We can easily see that $R$ satisfies $(b2)$ and $(J1)$.  Now $x\in R(u,v), y\in R(u,x)$, but $x\notin R(y,v)$.  So that $R$ fails to satisfy $(b3)$ axiom.
\end{example}
The following examples establish that the axioms $(J2), (J2'), (J3'), (b3)$ and $(b2)$ are independent.

\begin{example}[$(J2), (J2'), (J3'),(b3)$ but not $(b2)$]$~$\\
	Let $V=\{a,b,c,d,e\}$.  Define a transit function $R$ on $V$ as follows: $R(a,b)=\{a,b\}, R(a,c)=\{a,b,c \},R(a,d)=\{a,b,c,d\}, R(a,e)=\{a,b,d,e\}, R(b,c)=\{b,c\},R(b,d)=\{b,c,d\}, R(b,e)=\{b,c,d,e\},R(c,d)=\{c,d\}, R(c,e)=\{c,d,e\}, R(d,e)=\{d,e\}$   We can see that $R$ does not satisfy $(b2)$ as $d\in R(a,e)$ but $R(a,d)\nsubseteq R(a,e)$.  But we can see that $R$ satisfies $(J2), (J2'), (J3'),(b3)$.	
\end{example}
\begin{example}[$(J2),(J2'), (J3'),(b2)$ but not $(b3)$]$~$\\
	Let $V=\{a,b,c,d,e\}$.  Define a transit function $R$ on $V$ as follows: $R(a,b)=\{a,b\},R(a,c)=\{a,c\},R(a,d)=\{a,b,c,d\},R(a,e)=V, R(b,c)=\{b,c\}, R(b,d)=\{b,d\}, R(b,e)=\{b,e\}, R(c,d)=\{c,d\}, R(c,e)=\{c,b,d,e\}, R(d,e)=\{d,e\}$.  We can see that $R$ satisfies $(J2),(J2'),(J3'),(b2)$.  Now $d\in R(a,e), b\in R(a,d)$ but we can see that $d\notin R(b,e)$, so $R$ doesnot satisfy the $(b3)$ axiom.
\end{example}
\begin{example}[$(J2), (J3'),(b2),(b3)$ but not $(J2')$]$~$\\
	Let $V=\{a,b,c,d,e\}$.  Define a transit function $R$ on $V$ as follows: $R(a,b)=\{a,b\},R(a,c)=\{a,c\},R(a,d)=\{a,b,c,d\},R(a,e)=\{a,b,e\}, R(b,c)=\{b,c\}, R(b,d)=\{b,d\}, R(b,e)=\{b,e\}, R(c,d)=\{c,d\}, R(c,e)=\{c,b,d,e\}, R(d,e)=\{d,e\}$.  We can see that $R$ satisfies $(J2),(J3'),(b2),(b3)$.  We have $c\in R(a,d), d\in R(c,e),R(a,c)=\{a,c\}, R(a,d)=\{a,d\}, R(d,e)=\{d,e\}$ but $c\notin R(a,e)=\{a,b,e\}$.  Hence $R$ does not satisfy $(J2')$. 
\end{example}
\begin{example}[$(J2'),(b2), (J3'),(b3)$ but not $(J2)$]$~$\\
	Let $V=\{a,b,c,d,e\}$.  Define a transit function $R$ on $V$ as follows: $R(a,b)=\{a,b,c\}, R(a,e)=\{a,e\},R(b,e)=\{b,e\}$ and for all other pair define $R(x,y)=\{x,y\}$.  We can see that $R$ satisfies $(J2'),(J3'),(b2),(b3)$.  But  $R(a,e)=\{a,e\}, R(b,e)=\{b,e\}$ and $e\notin R(a,b)=\{a,b,c\}$.  So $R$ does not satisfy $(J2)$.
\end{example}
\begin{example}[$(J2),(J2'),(b2),(b3)$ but not $(J3')$]$~$\\
	Let $V=\{u,v,w,x,y,z\}$. Define a transit function $R$ on $V$ as follows: $R(u,x)=\{u,x\}, R(u,z)=\{u,x,z\}, R(u,y)=V=R(x,v)=R(z,w), R(u,v)=\{u,w,y\}, R(u,w)=\{u,w\}, R(x,z)=\{x,z\}, R(x,y)=\{x,z,y\}, R(x,w)=\{x,w\}, R(z,y)=\{z,y\}, R(z,v)=\{z,v\}, R(y,v)=\{y,v\}, R(y,w)=\{y,w\},R(v,w)=\{v,w\}$.  We can see that $R$ satisfies $(J2),(J2'),(b2),(b3)$.   But $x\in R(u,y), y\in R(x,v), R(x,y)\neq \{x,y\}, R(u,v)\neq \{u,v\}$ but $x\notin R(u,v)$.  So $R$ does not satisfy $(J3')$.
\end{example}
 %( \textbf{Hi Prasanth!. Can you provide examples to show the independence of the axioms (J2), (J2'), (b3) and (b2).That is, three more examples})

We have already noted in the introductory section that axiom $(b3)$ implies axiom $(b1)$, for any transit function $R$. 

Therefore, we replace $(b1)$ by $(b3)$ in Theorem~\ref{induced} and using Lemma~\ref{HHD3fan} and Proposition~\ref{b2b3j1}, we can reformulate Theorem~\ref{induced} using a minimal set of independent axioms as 

\begin{theorem}\label{Dis-h}
 Let $R$ be a transit function on a non-empty finite set $V$ satisfying the axioms  $(b2), (b3), (J2), (J2')$ and $(J3')$ with underlying graph
$G_R$. Then $G_R$ is $HHD3$-fan -free (distance hereditary graph) and $R$ is precisely the induced path transit function of $G_R$. 
\end{theorem}

A distance hereditary graph $G$ is precisely the graph in which every induced path is a shortest path and hence both the induced path transit function and the interval function coincide in $G$.  Therefore we have that Theorem~\ref{Dis-h} also holds for the interval function of $G_R$. 
 That is, we have 

\begin{theorem}\label{Dis-h1}
Let $R$ be any transit function satisfying the axioms  $(b2), (b3), (J2), (J2')$ and $(J3')$ then $G_R$ is distance hereditary and $R$ coincides with the interval function $I$ of $G_R$.
\end{theorem}

Also note that since axiom $(J0)$ implies axiom $(J3')$ and $(J2')$, we can use the same ideas in the proof of Theorem~\ref{Ptolemaic} to prove  an independent proof for Theorem~\ref{Dis-h1}.
Finally we have the following theorem characterizing the interval function of a distance hereditary graph from Theorem~\ref{J2J3Prime} and Theorem~\ref{Dis-h1}

\begin{theorem}\label{DH_Ch}
Let $R$ be a transit function on the vertex set $V$ of a connected graph $G$. Then $G$ is a distance hereditary graph and $R$ coincides the interval function $I_G$ of $G$ if and only if $R$  satisfies the axioms  $(b2), (b3), (J2), (J2'), (J3')$ and the axiom $R(u,v) = \{u,v\}$ implies that $uv\in E(G)$. 
\end{theorem}

%\begin{theorem}\label{j0'j2}
% Let $R$ be a transit function defined on a non-empty set $V$.  If $R$ satisfies (J0') and (J2) then the underlying graph $G_R$ of $R$ has no isometric cycles, $C_n$- for $n\ge 4$.
%\end{theorem}
%\begin{proof}
%	Let $R$ be a transit function satisfying (J0') and (J2).  Let $G_R$ contains an induced cycle say $C_n=u_1u_2\ldots u_nu_1$.  Without loss of generality assume $C_n$ is the minimum isometric cycle (in the sense that length of the cycle is the smallest possible). We prove the theorem for every $k\ge 4$.\\
%	\noindent
%	Case: $k=4$.\\
%Now since $R(u_1,u_2)=\{u_1,u_2\}$ and $R(u_2,u_3)=\{u_2,u_3\}$, By (J2)-axiom we have $u_2\in R(u_1,u_3)$ in a similar fashion we can show that $u_3\in R(u_2,u_4)$.  Consider $u-1,u-2,u_3,u_4$.  It follows that  $R(u_1,u_3)\cap R(u_2,u_4) = \{u_2,u_3\} \subseteq \{u_2,u_3\}$
%Since $R$ satisfies (J0')-axiom we have $u_2\in R(u_1,u_4)$, which is a contradiction as $R(u_1,u_4)=\{u_1,u_4\}$.\\
%\noindent
%Case: $k=5$.\\
%As in the above case, we can see that $u_4\in R(u_5, u_3)$, this together with  $u_3\in R(u_4,u_1)$ we have $u_4\in R(u_5,u_1)$, which is again a contradiction. \\
%\noindent
%Case: $k\ge 6$.\\
%By repeated application of (J2)-axiom, as in the above two cases, $u_{n-1}\in R(u_n,u_{n-2})$ with $u_{n-2}\in R(u_{n-1},u_1)$, by applying (J2), we can see that $u_{n-1}\in R(u_n,u_1)$.  Here again a contradiction.\\

%Hence, in all cases, we can see that $G_R$ does not contain $C_n$ as an induced subgraphs.  This completes theorem.
%\end{proof}
\bibliographystyle{amsplain}
\bibliography{bridged}
\end{document}